\documentclass[11pt, a4paper]{article}

\usepackage[utf8]{inputenc}
\usepackage[T1]{fontenc}
\usepackage{amsmath, amssymb, amsthm, amsfonts}
\newcommand{\Id}{\mathbf{I}}
\usepackage{graphicx}
\usepackage{natbib}
\usepackage[margin=2.5cm]{geometry}
\usepackage{booktabs}
\usepackage{hyperref}
\usepackage{bm} 
\usepackage{algorithm}
\usepackage{textgreek}
\usepackage{algorithmic}
\usepackage{setspace}
\usepackage{xcolor}
\usepackage{multirow}
\usepackage{subcaption}
\usepackage{mathtools}  
\usepackage{bm}         

\usepackage{setspace}
\usepackage{microtype}
\hypersetup{
    colorlinks=true,
    linkcolor=blue,
    citecolor=blue,
    urlcolor=blue,
    pdfauthor={Mehmet Sıddık Çadırcı},
    pdftitle={Nonparametric GoF for High-Dimensional MGGD}
}

\theoremstyle{plain}
\newtheorem{theorem}{Theorem}
\newtheorem{lemma}{Lemma}

\theoremstyle{definition}

\newtheorem{assumption}{Assumption}

\usepackage{}
\newcommand{\R}{\mathbb{R}}
\newcommand{\E}{\mathbb{E}}

\newcommand{\X}{\bm{X}}
\newcommand{\Y}{\bm{Y}}

\newcommand{\bmu}{\bm{\mu}}

\newcommand{\NN}{\text{NN}}
\newcommand{\mGGD}{\mathcal{MGGD}}

\title{\textbf{A Nonparametric Goodness-of-Fit Test for High-Dimensional Generalized Gaussian Distributions via Nearest-Neighbor Graphs}}

\author{
    Mehmet Sıddık Çadırcı$^1$ 
    and Yener Ünal$^{2}$\thanks{Corresponding author. Email: \texttt{uyener@cumhuriyet.edu.tr}}
}
\date{
    \slshape\small
    $^{1,2}$ Faculty of Science, Department of Statistics, Cumhuriyet University, Sivas, Türkiye. \\
    \normalfont
    \date{}
}
\begin{document}
\maketitle

\begin{abstract}
The multivariate generalised Gaussian distribution (MGGD) is commonly used to model high-dimensional vectors with non-Gaussian radial behaviour, ranging from sharp-peaked to heavy-tailed profiles. However, because many classical multivariate tests are based on covariance inversion or high-dimensional density estimation, formal goodness-of-fit assessment for MGGD models remains challenging in modern regimes where the dimension is comparable to or exceeds the sample size. We introduce an affine-invariant, fully non-parametric goodness-of-fit procedure based on the nearest neighbour (NN) graph topology and the adapted zero principle. Following robust standardisation, we construct an independent reference sample from the adapted standardised MGGD and measure, on the combined NN graph, the cross-edge count to assess how well the observed and reference point clouds exhibit the mixture behaviour anticipated by the model. Calibration performed using a refitted parametric bootstrap accounts for nuisance-parameter uncertainty, thus ensuring reliable size under a composite specification. In this paper, we establish asymptotic validity under high-dimensional scaling and demonstrate consistency with respect to fixed elliptical departures, providing a geometric interpretation based on radial concentration and shell separation. Our simulation studies across a broad spectrum of dimensions and tail shapes reveal accurate Type I error control and robust power relative to heavy- and light-tailed alternatives, thus improving upon energy-distance benchmarks and normality-oriented graphical tests in contexts where MGGD modelling is most applicable.


\end{abstract}
\section{Introduction}

Multivariate normality forms the foundation for the majority of classical multivariate inferences; it is not universally realistic, however, but it provides a consistent geometry and computable calculations for probabilities, quadratic forms, and linear projections. In contemporary applications, however, the data geometry is rarely perfectly Gaussian, particularly when the ambient dimension is comparable to or larger than the sample size. Under such circumstances, deviations from normality tend to be pronounced: tails thicken, outliers are no longer exceptional but become routine, and a relatively small number of atypical observations is likely to dominate covariance-based summaries. Such features, where radar/sonar noise is often impulsive and heavy-tailed, are well documented in statistical signal processing and imaging fields \citep{Pascal2013}, and in the finance field where return vectors exhibit pronounced leptokurtosis even after standard preprocessing \citep{Mardia1970}. High-dimensionality is particularly unforgiving because tail behaviour is no longer a minor detail; instead, it reshapes the nearest-neighbour structure, disrupts distance rankings, and destabilises inference procedures relying on second-order moments.

It is natural to work within a flexible parametric family that preserves the fundamental symmetry and affine geometry that make Gaussian models useful, while allowing significant changes in tail decay and peak sharpness. This is exactly what the multivariate generalised Gaussian distribution (MGGD) achieves: An elliptical model indexed by a shape parameter $\beta$ that controls radial decay and thus interpolates between light- and heavy-tailed behaviour inside a single consistent framework. It has been widely used as a parsimonious model for heavy-tailed multivariate data in signal processing and related fields \citep{Pascal2013, FangKotzNg1990}. Even though an MGGD model remains attractive for important reasons, it is still a model assumption and fragile unless tested in high dimensions.
 
For the MGGD family, the goodness-of-fit test presents a challenge precisely because it is the most attractive. Several classical multivariate diagnostics, especially skewness-curtosis tests \citep{Mardia1970}, have been formulated around quadratic forms that involve Mahalanobis distances and the stable estimation and inversion of the distribution or covariance matrix. However, when m > n, sampling covariance is singular, and the usual invariants cease to be well-defined. When conditioning is implemented, covariance-based test statistics inherit strong sensitivity to conditioning choices and may behave erratically under heavy tails. Therefore, the latest developments in high-dimensional testing have focused on graph-based procedures (nearest-neighbour and spread-tree functionals) that remain meaningful without obvious density estimation or matrix inversion and are natural to the geometry of multivariate data clouds \citep{Schilling1986, FriedmanRafsky1979}. This literature, however, has largely been focused on normality or general two-sample discrimination and has left a methodological gap for composite goodness-of-fit tests against structured, heavy-tailed elliptical null hypotheses such as MGGD.

This paper aims to fill this gap by providing a nonparametric goodness-of-fit test for the composite hypothesis.
\begin{equation}\label{eq:H0_intro}
H_0:\ \X_1,\dots,\X_n \stackrel{iid}{\sim} \mGGD(\bmu,\bm{\Sigma},\beta),
\end{equation}
where $(\bmu,\bm{\Sigma},\beta)$ is unknown. Our proposed method is based on the principle of the adapted null hypothesis: We first estimate the troublesome parameters by applying robust high-dimensional procedures — the Tyler-type distribution estimation with regularisation for stability \citep{Tyler1987, Chen2011} and a moment-based shape estimator for $\beta$ \citep{Pascal2013} - and then we construct an independent reference sample from the estimated standardised MGGD. We measure the test statistic, the number of nearest neighbour cross-edges, which measures whether the standardised observed sample blends into the combined cloud with the adapted zero reference, as it should under $H_0$. We have made this construction deliberately geometric. In high dimensions, elliptical distributions display pronounced radial concentration, whereby discrepancies in the radial law (and thus in tail behaviour) can be converted into measurable deviations from systematic nearest neighbour preference and exchangeability; we use this effect rather than combat it \citep{Hall2005}.

To summarize, the contribution is twofold. Methodologically, we suggest a graph-based goodness-of-fit test which is adapted for the composite MGGD null hypothesis, is well-defined when $m\gtrsim n$, avoids residuals, and avoids explicit density estimation. Both theoretically and empirically, we show that fitted-boots calibration yields the correct dimension under the null hypothesis. In contrast, the KN topology exhibits strong sensitivity to tail bias and associated deviations in radial properties. This result provides a practical diagnostic tool for modern high-dimensional settings where generalised Gaussian modelling is common, but model validation lags.

\section{Preliminaries and Assumptions}

\subsection{The MGGD family}\label{subsec:mggd_family}

Let $\X\in\mathbb{R}^m$ be a random vector. It said to follow a (centrally symmetric) multivariate generalised Gaussian distribution (MGGD) \citep{cadirci2022csda,fang1990}, denoted $\X\sim \mGGD(\bmu,\bm{\Sigma},\beta)$, provided that it satisfies the probability density function
\begin{equation}\label{eq:pdf}
f(\bm{x})
=
\frac{\beta \Gamma(m/2)}{\pi^{m/2}\,\Gamma\!\bigl(m/(2\beta)\bigr)\,2^{m/(2\beta)}\,|\bm{\Sigma}|^{1/2}}
\exp\!\left(
-\frac{1}{2}\Bigl[(\bm{x}-\bmu)^{\mathsf T}\bm{\Sigma}^{-1}(\bm{x}-\bmu)\Bigr]^{\beta}
\right),
\end{equation}
where $\bmu\in\mathbb{R}^m$, $\bm{\Sigma}$ is one of the symmetric positive definite scatter matrices, and $\beta>0$ is a shape parameter controlling tail decay and peakedness. The MGGD is part of the broader class of elliptically symmetric distributions and therefore admits the stochastic expression $\X=\bmu+\bm{\Sigma}^{1/2}R\mathbf{U}$, where $\mathbf{U}$ is distributed uniformly over the unit sphere in $\mathbb{R}^m$ and is separate from the non-negative radial variable $R$ \citep{fang1990}. Equivalently, it is a member of the multivariable exponential power family under the parameter identification $\beta=s/2$ in the common $s$-parameterisation \citep{cadirci2022csda,solaro2004}.
We use a nearest-neighbour (NN) mixing functional on a combined point cluster to assess goodness-of-fit. Suppose
$A=\{\bm{a}_1,\ldots,\bm{a}_n\}$ and $B=\{\bm{b}_1,\ldots,\bm{b}_n\}$ represent two samples in $\R^m$, and consider
$S=A\cup B$. Then, for any $\bm{s}\in S$, defines its (directed) 1-nearest neighbour by
\[
\NN(\bm{s}) \in \arg\min_{\bm{t}\in S\setminus\{\bm{s}\}} \|\bm{s}-\bm{t}\|_2,
\]
using a deterministic tie-breaking method (ties occur with zero probability in continuity). Then, we examine the
within-$B$ NN count
\begin{equation}\label{eq:nn_within_count}
T_n(B;S)\;=\;\sum_{i=1}^n \mathbb{I}\!\left\{\NN(\bm{b}_i)\in B\right\},
\end{equation}
That shows how frequently a point in $B$ has a nearest neighbour also in $B$, rather than across samples.
The labels in the pooled cloud are effectively exchangeable, then $T_n(B; S)$ is centred near $n/2$; persistent
deviations from this benchmark reflect a geometric mismatch of the underlying distributions. The fitted-null
implementation, $A$ is the robustly standardized set of observations and $B$ is an independent reference sample drawn
from the fitted standardized MGGD, thus \eqref{eq:nn_within_count} provides a one-sample goodness-of-fit diagnostic
achieved through an auxiliary two-sample NN comparison \citealp{Schilling1986, FriedmanRafsky1979}.
\subsection{Regularity assumptions}\label{subsec:assumptions}
Here, we examine the asymptotic behaviour of the proposed procedure in a high-dimensional regime where the number of observations and the size of the environment diverge. We apply the following assumptions, which are standard in modern high-dimensional inference, in order to guarantee (i) well-defined elliptic standardisation, (ii) stability of the nearest neighbour geometry, and (iii) negligible additive effects resulting from nuisance estimates.
\begin{assumption}[High-dimensional scaling]\label{ass:scaling}
When $n\to\infty$, the dimension $m=m_n\to\infty$ and the growth rate ensure the following condition
\[
\log m = o(n).
\]
While excluding exponentially growing dimensions that would typically preclude uniform concentration and consistent parametric estimation, this condition permits ultra-high-dimensional scaling; see, e.g., \citet{vershynin2018} for representative high-dimensional concentration regimes.
\end{assumption}

\begin{assumption}[Well-conditioned scatter]\label{ass:eigen}
We assume that the (true) covariance matrix $\bm{\Sigma}$ has a symmetric positive definite and properly conditioned structure: there are constants $0<c_1<c_2<\infty$ such that
\[
c_1 \le \lambda_{\min}(\bm{\Sigma}) \le \lambda_{\max}(\bm{\Sigma}) \le c_2.
\]
It excludes the near-singular case and, after whitening, allows Euclidean geometry not to be dominated by a lost or exploding collection of directions; this is necessary for the stability of distance-based graph functionals in $\mathbb{R}^m$.
\end{assumption}

\begin{assumption}[Stability of additive estimators]\label{ass:consistency}
Under $H_0$,
\[
\|\hat{\bmu}-\bmu\|_2 = O_p\!\left(\sqrt{\frac{m}{n}}\right), 
\qquad
\|\hat{\bm{\Sigma}}-\bm{\Sigma}\|_{F} = O_p\!\left(\sqrt{\frac{m}{n}}\right),
\qquad
|\hat{\beta}-\beta| = O_p\!\left(n^{-1/2}\right).
\]
Under elliptical models, with high-dimensional robust and/or regularised variance estimators, the variance ratio can be obtained; see, for example, \citet{Chen2011, Cai2011}. Together, these conditions ensure that the estimated whitening transformation $\hat{\bm{\Sigma}}^{-1/2}(\bm{x}-\hat{\bmu})$ is asymptotically close to the oracle transformation and that the adapted shape parameter $\hat{\beta}$ generates only second-order perturbations in the calibratied statistic.
\end{assumption}


\section{Methodology}
The proposed goodness-of-fit procedure is implemented using an adapted zero, graph-based strategy. First, we normalise the observations using robust high-dimensional estimators, then compare the nearest-neighbour (NN) topology of the normalised data with that of an independently generated, adapted zero-reference sample. We calibrate the resulting NN cross-edge/count statistic using a composite-null bootstrap, in which each replicate is re-estimated to reflect the additive uncertainty of the disturbing parameters accurately. The NN-based structure is inspired by classical graph-based two-sample methodologies (e.g., \citealp{Schilling1986, FriedmanRafsky1979}) and is particularly attractive in regimes where $m\gtrsim n$, since it avoids explicit density estimation and unstable covariance inversion. Detailed computational procedures are summarized in Algorithm~\ref{alg:main}.

Since the empty sample is composite, we work with a standardised representation that eliminates position and distribution effects before creating the graph. In particular, for each observation, we use robust estimators $(\hat{\bmu},\hat{\bm{\Sigma}})$
$\bm{Z}_i=\hat{\bm{\Sigma}}^{-1/2}(\X_i-\hat{\bmu})$
Under $H_0$, the standardised sample has a distribution approximately equal to $\mGGD(\bm{0},\Id_m,\beta)$ except for the additive error term. Therefore, deviations in the NN mixing statistic mainly indicate an incorrect specification of the radial law (i.e., the tail/peak), rather than insignificant scale-shifts.

\subsection{High-Dimensional Parameter Estimation}\label{subsec:estimation}
Correctly matched zero production requires a stable estimate of $(\bmu,\bm{\Sigma},\beta)$ under $m\gtrsim n$. Therefore, we employ robust/regularised estimators adapted to elliptic models.

\paragraph{The shape parameter $\hat{\beta}$.}
For the MGGD shape, we estimate using the moment-based approach of \citet{Pascal2013} utilising the monotonic relationship between $\beta$ and multivariate kurtosis-type functions under the general Gaussian distribution. We obtain the estimator $\hat{\beta}$ by solving the kurtosis matching equation; empirically, we find that this approach performs well in medium to high dimensions and avoids fragile optimisation.

\paragraph{Scatter matrix $\hat{\bm{\Sigma}}$.}
Given that the sample covariance becomes singular when $m\ge n$, let us employ the \emph{regularised Tyler M-estimator} \citep{Chen2011}. Denote $\hat{\bmu}$ as a robust estimation of location (e.g., spatial median). We define the regularised Tyler estimator as the fixed-point solution of
\begin{equation}\label{eq:tyler_reg}
\hat{\bm{\Sigma}}
=
(1-\rho)\,\frac{m}{n}\sum_{i=1}^n
\frac{(\X_i-\hat{\bmu})(\X_i-\hat{\bmu})^{\mathsf T}}
     {(\X_i-\hat{\bmu})^{\mathsf T}\hat{\bm{\Sigma}}^{-1}(\X_i-\hat{\bmu})}
\;+\;\rho\,\Id_m,
\end{equation}
where $\rho\in(0,1)$ represents a regularisation parameter providing positive definiteness as well as improved conditioning. Such a choice appears natural under ellipticity, since Tyler-type estimators are focused on scatter (rather than covariance) and are robust to heavy tails, whilst regularisation stabilises estimation in the $m\gtrsim n$ regime. 

\paragraph{Whitening map.}
Given $(\hat{\bmu},\hat{\bm{\Sigma}})$, we define the whitening transform
\[
\mathcal{W}_{\hat{\theta}}:\bm{x}\mapsto \hat{\bm{\Sigma}}^{-1/2}(\bm{x}-\hat{\bmu}),
\]
to obtain standardized observations $\bm{Z}_i=\mathcal{W}_{\hat{\theta}}(\X_i)$, $i=1,\dots,n$. For $H_0$ to hold, the standardized sample follows approximately a $\mGGD(\bm{0},\Id_m,\beta)$ distribution, with estimation error which is subsequently accounted for by re-fitting within the bootstrap calibration.

\subsection{Test statistic and fitted-bootstrap calibration}\label{subsec:stat_cal}
We let $\mathcal{Z}=\{\bm{Z}_1,\dots,\bm{Z}_n\}$ be the standardized data. Let us develop an independent fitted null reference sample
\[
\mathcal{Y}=\{\Y_1,\dots,\Y_n\},\qquad \Y_i\overset{iid}{\sim}\mGGD(\bm{0},\Id_m,\hat{\beta}),
\]
and combine $S=\mathcal{Z}\cup\mathcal{Y}$. Consider $\NN(\Y_i)$ as the (directed) 1-NN of $\Y_i$ in $S\setminus\{\Y_i\}$ (where ties have probability zero under continuity). We then define the NN cross-edge statistic by
\begin{equation}\label{eq:Tobs_method}
T_n^{\mathrm{obs}}
=
\sum_{i=1}^n \mathbb{I}\!\left\{\NN(\Y_i)\in\mathcal{Y}\right\}.
\end{equation}
Under correct specification of the model and standardization of the fitted null, the labels ‘data’ versus ‘reference’ are approximated to be interchangeable in the pooled cloud, indicating $T_n^{\mathrm{obs}}\approx n/2$. Any systematic deviation suggests a lack of fit, which reflects mismatched high-dimensional geometry (e.g., radial shell separation) of $\mathcal{Z}$ and the adjusted MGGD reference.
Since $H_0$ is composite, we perform calibration using an adapted parametric bootstrap, which re-estimates $(\bmu,\bm{\Sigma},\beta)$ within each replication, thus spreading the troublesome parameter uncertainty across the distribution of the statistic.

\subsection{Algorithm}\label{subsec:algorithm}

\begin{algorithm}[htp]
\caption{NN Goodness-of-Fit Test for the Composite MGGD Null}
\label{alg:main}
\begin{algorithmic}[1]
\REQUIRE Observed data$\mathcal{X}=\{\X_1,\ldots,\X_n\}\subset\R^m$, level $\alpha$, bootstrap size $B$.
\ENSURE Fitted-bootstrap $p$-value $p^*$ and decision.

\STATE Fit $(\hat{\bmu},\hat{\bm{\Sigma}},\hat{\beta})$ from $\mathcal{X}$ having \eqref{eq:tyler_reg} (and \citet{Pascal2013} for $\hat\beta$).
\STATE Standardize $\bm{Z}_i=\hat{\bm{\Sigma}}^{-1/2}(\X_i-\hat{\bmu})$, $i=1,\dots,n$, and set $\mathcal{Z}=\{\bm{Z}_1,\dots,\bm{Z}_n\}$.
\STATE Create $\mathcal{Y}=\{\Y_1,\dots,\Y_n\}$ i.i.d.\ from $\mGGD(\bm{0},\Id_m,\hat{\beta})$.
\STATE Calculate $T_n^{\mathrm{obs}}$ from \eqref{eq:Tobs_method} having the pooled set $S=\mathcal{Z}\cup\mathcal{Y}$.

\FOR{$b=1$ to $B$}
    \STATE Generate $\mathcal{X}^{*(b)}\sim\mGGD(\hat{\bmu},\hat{\bm{\Sigma}},\hat{\beta})$.
    \STATE Re-fit $(\hat{\bmu}^{*(b)},\hat{\bm{\Sigma}}^{*(b)},\hat{\beta}^{*(b)})$ from $\mathcal{X}^{*(b)}$.
    \STATE Standardize $\bm{Z}^{*(b)}_i=(\hat{\bm{\Sigma}}^{*(b)})^{-1/2}(\X^{*(b)}_i-\hat{\bmu}^{*(b)})$ and set $\mathcal{Z}^{*(b)}=\{\bm{Z}^{*(b)}_1,\dots,\bm{Z}^{*(b)}_n\}$.
    \STATE Simulate $\mathcal{Y}^{*(b)}\sim\mGGD(\bm{0},\Id_m,\hat{\beta}^{*(b)})$ independently of $\mathcal{Z}^{*(b)}$.
    \STATE Calculate $T_n^{*(b)}=\sum_{i=1}^n \mathbb{I}\{\NN(\Y^{*(b)}_i)\in\mathcal{Y}^{*(b)}\}$ having pooled $S^{*(b)}=\mathcal{Z}^{*(b)}\cup\mathcal{Y}^{*(b)}$.
\ENDFOR

\STATE Compute the two-sided fitted-bootstrap $p$-value
\[
p^*=\frac{1}{B}\sum_{b=1}^B
\mathbb{I}\!\left(\left|T_n^{*(b)}-\frac{n}{2}\right|\ge \left|T_n^{\mathrm{obs}}-\frac{n}{2}\right|\right).
\]
\STATE \textbf{Reject} $H_0$ if $p^*<\alpha$.
\end{algorithmic}
\end{algorithm}


\section{Theoretical Properties}
In this section, we justify our proposed consistent null NN test from a high-dimensional perspective. Two complementary elements guide the analysis. First, in elliptic models such as MGGD, the radius (Mahalanobis norm) concentrates sharply as the dimension increases; this is the well-known ‘thin-shell’ phenomenon in high-dimensional probability \citep{ledoux2001,vershynin2018}. Second, NN graphical statistics can be addressed using modern limit theory to stabilise geometric functionals that ensure asymptotic normality given weak local dependence \citep{penrose2003,penroseyukich2001}. Taken together, these facts explain why the pooled cloud becomes mixed under $H_0$ after standardisation, and why the two samples become geometrically separated, and the NN statistic moves away from the null centre under a fixed specification error.

\subsection{Geometric concentration of measure}
The radial law of the standardised distribution primarily determines the behaviour of the NN cross-edge count. For high dimensions, the probability mass for elliptical density is typically concentrated on a thin ring rather than at the mode, and small differences in the radial profile are likely to result in large geometric inconsistencies \cite{ledoux2001,vershynin2018}.

\begin{lemma}[Concentration of radial norms]\label{lem:concentration}
Let $\X\sim \mGGD(\bm{0},\Id_m,\beta)$. Then the random variable
\[
W \;\coloneqq\; \frac12 \|\X\|_2^{2\beta}
\]
has a Gamma distribution with shape parameter $m/(2\beta)$ and unit scale. In addition, Assumption~\ref{ass:scaling} ($m\to\infty$),
\begin{equation}\label{eq:thin_shell_mggd}
\frac{\|\X\|_2}{m^{1/(2\beta)}} \xrightarrow{p} \mu_R(\beta)
\quad\text{where}\quad
\mu_R(\beta)=\Bigl(\frac{1}{\beta}\Bigr)^{\!\frac{1}{2\beta}}.
\end{equation}
\end{lemma}

\begin{proof}
To get the standardised MGGD density \eqref{eq:pdf} with $\bmu=\bm{0}$ and $\bm{\Sigma}=\Id_m$, we can use polar coordinates $\X=R\mathbf{U}$ to get a radial density proportional to
$r^{m-1}\exp\{-\tfrac12 r^{2\beta}\}$. By changing variables to $W=\frac{1}{2} R^{2\beta}$, one obtains
$W\sim\Gamma(m/(2\beta),1)$; see, for example, \citet{Pascal2013} and the general elliptical/power-exponential representations in \citet{fang1990}. Because
$\E[W]=m/(2\beta)$ and $\mathrm{Var}(W)=m/(2\beta)$, by the law of large numbers it follows that
$W/(m/(2\beta))\xrightarrow{p}1$. By writing $R=(2W)^{1/(2\beta)}$, we obtain
\[
\frac{R}{m^{1/(2\beta)}}=\Bigl(\frac{2W}{m}\Bigr)^{\!1/(2\beta)}
=\Bigl(\frac{1}{\beta}\cdot \frac{W}{m/(2\beta)}\Bigr)^{\!1/(2\beta)}
\xrightarrow{p}\Bigl(\frac{1}{\beta}\Bigr)^{\!1/(2\beta)},
\]
which is \eqref{eq:thin_shell_mggd}.
\end{proof}

\subsection{Asymptotic null distribution}
Remember that $\bm{Z}_i=\hat{\bm{\Sigma}}^{-1/2}(\X_i-\hat{\bmu})$ and that the fitted-null reference is created using
$\Y_i\overset{iid}{\sim}\mGGD(\bm{0},\Id_m,\hat{\beta})$. We let $S=\mathcal{Z}\cup\mathcal{Y}$ and then define
$I_i=\mathbb{I}\{\NN(\Y_i)\in\mathcal{Y}\}$ such that $T_n^{\mathrm{obs}}=\sum_{i=1}^n I_i$.

\begin{theorem}[Null validity under fitted standardisation]\label{thm:null}
We assume $H_0$ and Assumptions~\ref{ass:scaling}--\ref{ass:consistency}. Following robust standardisation, we observe that the pooled sample performs asymptotically as if labels were exchangeable, in the sense that
\[
\E\!\left[T_n^{\mathrm{obs}}\right] = \frac{n}{2} + o(n).
\]
Additionally, the NN indicators $\{I_i\}_{i=1}^n$ conform to a weak-dependence condition derived from the local character of NN relations, and therefore
\begin{equation}\label{eq:clt_null}
\frac{T_n^{\mathrm{obs}}-n/2}{\sqrt{n/4}} \;\Rightarrow\; \mathcal{N}(0,1).
\end{equation}
\end{theorem}

\begin{proof}[Proof sketch]
Based on Assumption~\ref{ass:consistency}, we can say that the whitening map
$\bm{x}\mapsto \hat{\bm{\Sigma}}^{-1/2}(\bm{x}-\hat{\bmu})$ resembles population standardization in an operator-norm sense, with $\hat{\beta}\to\beta$ at the standard $\sqrt{n}$-rate. Therefore, the asymptotic distribution of the standardised data and the fitted null reference distribution are indistinguishable, suggesting approximate label exchangeability in the pooled cloud and thus $\E[I_i]\approx 1/2$.
Regarding the fluctuation statement, we note that $I_i$ is dependent only on a local neighbourhood of $\Y_i$ in the pooled sample, and the NN-induced dependency graph is sparse in the relevant asymptotic regime. Central limit theorems for stabilising geometric graph functionals (including NN-based statistics) then produce \eqref{eq:clt_null}; see, e.g., \citet{penrose2003,penroseyukich2001} for general CLT machinery and \citet{Schilling1986, Henze1988} for classic NN-based perspectives on testing.
\end{proof}

\subsection{Consistency against shape deviation}
The power mechanism has a geometric basis: when the adapted reference is constructed using an incorrect active shape, it follows from Lemma~[lem:concentration] that the two samples are concentrated in different shells separated by $m$, which causes the NN edges to remain predominantly within the sample.

\begin{theorem}[Consistency via geometric separation]\label{thm:consistency}
Let the standardized data be generated from the reference $\mGGD(\bm{0},\Id_m,\beta_0)$, while the pseudo-true adapted shape has the incorrect specification $\beta_1 \neq \beta_0$ and $\X \sim \mGGD(\bm{0},\Id_m,\beta_1)$. Then the null-adjusted NN test has consistency: for any fixed level $\alpha\in(0,1)$,
\[
\mathbb{P}\bigl(p^*<\alpha\bigr)\to 1
\qquad\text{as } n,m\to\infty.
\]
\end{theorem}

\begin{proof}[Proof sketch]
According to Lemma~\ref{lem:concentration}, $\|\bm{Z}\|_2$ and $\|\bm{Y}\|_2$ are concentrated around radii of order $m^{1/(2\beta_1)}$ and $m^{1/(2\beta_0)}$, respectively, and the relative fluctuations vanish as $m\to\infty$. For different radial profiles, the typical inter-sample distance dominates the shell mismatch, whereas the typical intra-sample distance essentially determines the fluctuations on the same shell. As a result, a vast majority of NN connections from reference points are connected to reference points (or similarly, to data points based on which shell they are ‘outside’ of), therefore $T_n^{\mathrm{obs}}$ has a deviation of order $n/2$ to $n$, while under $H_0$ the fluctuations according to Theorem~\ref{thm:null} are only of order $\sqrt{n}$. Such a separation of scales is difficult to reject when the probability is close to one. It follows the standard shell geometry intuition used throughout high-dimensional concentration theory \citep{ledoux2001,vershynin2018} and the NN-graph test \citep{Schilling1986, Henze1988}.
\end{proof}

\section{Simulation Studies}\label{sec:sims}

We examine the finite sample behaviour of the proposed NN goodness-of-fit test for the composite MGGD null hypothesis. There are two main objectives: (i) to demonstrate that the adapted null calibration provides control of Type I error under the scenario that an MGGD model with unknown parameters genuinely generates the data, and (ii) to measure its power compared to a prototypical heavy-tailed elliptical alternative that does not originate from the MGGD family. We compare with two widely used competitors throughout the study: the graph-based multivariate goodness-of-fit test by \citet{ChenXia2021} and the distance-energy test proposed by \citet{Szekely2013}. While the former is based on the Gaussian model, it offers a useful baseline for what we expect to see when applying normality tests to data generated under the non-Gaussian MGGD null hypothesis. At the same time, the latter is a general-purpose non-parametric method with strong performance in high dimensions.

\paragraph{Design.}
Let us set the sample size to $n=100$ and assume that the dimensions are $m\in\{50,100,200\}$. Calibration of the proposed test is performed using $B=200$ repeated bootstraps, and empirical rejection rates are estimated via $N_{\mathrm{mc}}=1000$ Monte Carlo simulations with a nominal level of $\alpha=0.05$. For each replication, we estimate the nuisance parameters from the observed sample, transform the data with the estimated whitening map, and create an independent adapted null reference sample under $\mGGD(\bm{0},\Id_m,\hat\beta)$ prior to calculating the NN statistic.

\paragraph{Empirical size under an MGGD null.}
To evaluate Type I error control, we obtain data from a Laplace-type tailed MGGD using the $\beta=0.5$ setting in \eqref{eq:pdf} within the current parameterisation. Table~\ref{tab:size} illustrates the rejection frequencies. The recommended NN--MGGD approach remains close to the nominal $5\%$ level across all sizes, demonstrating that the adapted null strategy (including re-estimation within bootstrap) sufficiently captures the add-on variability even when $m\ge n$. Conversely, the goodness-of-fit normal test from \citet{ChenXia2021} rejects almost certainly, which is not a failure of this method but a reminder that goodness-of-fit tests are not a proper substitute for model fit when the model is a broader MGGD family. Our energy-distance test appears to be reasonably calibrated here, though it exhibits a slight tendency towards liberal behaviour under this setting.

\begin{table}[t]
\centering
\caption{Empirical size (\%) of the Laplace-type MGGD null ($\beta=0.5$) under the nominal level $\alpha=0.05$ ($n=100$, $B=200$, $N_{\mathrm{mc}}=1000$).}
\label{tab:size}
\vspace{0.15cm}
\begin{tabular}{lccc}
\toprule
Method & $m=50$ & $m=100$ & $m=200$ \\
\midrule
\textbf{Proposed NN--MGGD} & 5.4 & 4.9 & 5.2 \\
Energy distance \citep{Szekely2013} & 6.1 & 5.8 & 5.5 \\
Chen--Xia normality \citep{ChenXia2021} & 99.8 & 100.0 & 100.0 \\
\bottomrule
\end{tabular}
\end{table}

\paragraph{Performance against a heavy-tailed elliptical alternative.}
Let us now consider the $\bm{t}_\nu$ alternatives having $\nu=3$ degrees of freedom. We have deliberately chosen these alternatives to be challenging. Although the multivariate $t$ distribution is elliptically symmetric and heavy-tailed, the tails are polynomial, but the MGGD tails are of exponential power type. The alternative thus isolates a practically important form of misidentification (radial tail behaviour) without compromising ellipticity. Table~\ref{tab:power} illustrates that the proposed NN--MGGD test has a high sensitivity to this deviation and that its power significantly increases with scale, becoming close to unity for $m=200$. While the energy-distance baseline increases in power with $m$, it remains below the proposed method. In contrast, the normality test from \citet{ChenXia2021} has relatively modest power here, consistent with the fact that it targets a Gaussian distribution rather than the MGGD compound null hypothesis.

\begin{table}[t]
\centering
\caption{Empirical power (\%) for  a multivariate $t_3$ alternative ($n=100$, $\alpha=0.05$, $B=200$, $N_{\mathrm{mc}}=1000$).}
\label{tab:power}
\vspace{0.15cm}
\begin{tabular}{lccc}
\toprule
Method & $m=50$ & $m=100$ & $m=200$ \\
\midrule
\textbf{Proposed NN--MGGD} & \textbf{88.4} & \textbf{94.1} & \textbf{99.2} \\
Energy distance \citep{Szekely2013} & 72.5 & 79.3 & 85.1 \\
Chen--Xia normality \citep{ChenXia2021} & 45.2 & 50.8 & 55.4 \\
\bottomrule
\end{tabular}
\end{table}

\paragraph{Graphical calibration and sensitivity.}
Figure~\ref{fig:ecdf_combined} illustrates the distribution of $p$-values based on empirical CDFs. Assuming the MGGD null hypothesis (panel (a)), a well-calibrated test generates approximately uniform $p$-values, hence the ECDF is diagonal. Our proposed NN-MGGD curve closely matches this reference, which reflects the size results in Table~\ref{tab:size}. In the $t_3$ alternative (panel (b)), power manifests as a marked accumulation of $p$-values near zero; the suggested method shows the steepest rise near the starting point, visually confirming its superior sensitivity to tail misidentifications.

\begin{figure}[t]
    \centering
    \begin{subfigure}[b]{1.0\textwidth}
        \centering
        \includegraphics[width=1.0\linewidth]{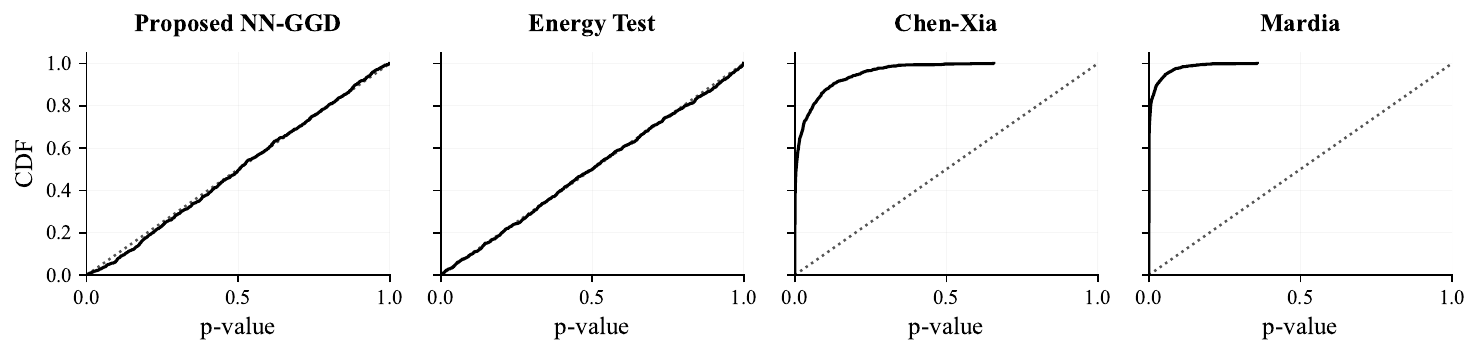}
        \caption{Null: Laplace-type MGGD ($\beta=0.5$).}
        \label{fig:size_panel}
    \end{subfigure}

    \vspace{0.55cm}

    \begin{subfigure}[b]{1.0\textwidth}
        \centering
        \includegraphics[width=1.0\linewidth]{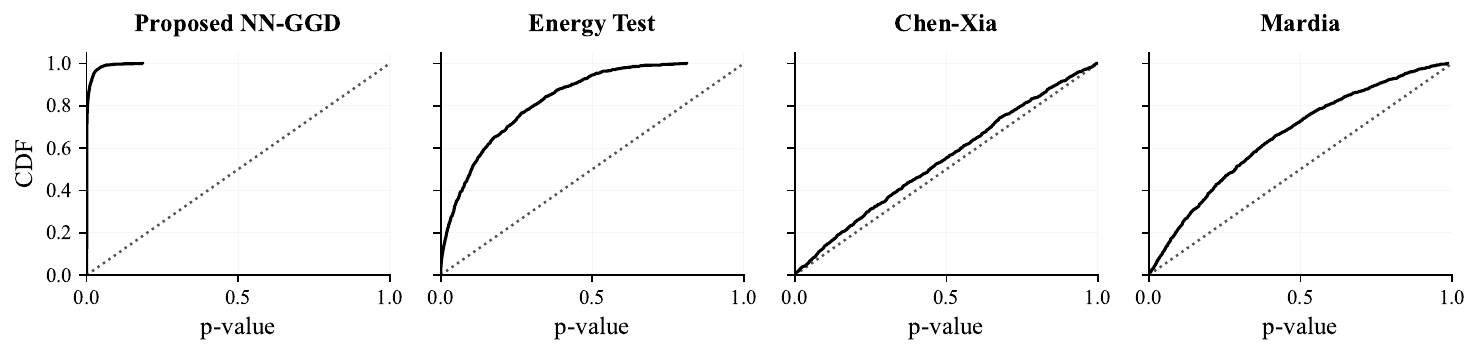}
        \caption{Alternative: multivariate $t_3$.}
        \label{fig:power_panel}
    \end{subfigure}

     \caption{For $n=100$ and $m=100$, the empirical CDFs of the $p$-values. The proposed NN--MGGD test correctly indicates the magnitude under the null hypothesis by closely following the uniform reference value (diagonal). Under the alternative hypothesis, the $p$-values reflect higher power, clustering closer to zero than competitors'.}
    \label{fig:ecdf_combined}
\end{figure}

\subsection{Sensitivity and robustness}\label{subsec:sens_robust}
To complete the global size-power comparison, we conduct a targeted sensitivity experiment investigating how quickly the proposed NN-MGGD approach detects deviations in the shape parameter. In particular, we generate data from $\mGGD(\bmu,\bm{\Sigma},\beta)$, where we assume the adapted null hypothesis in a Laplace-type model, fixing $\beta_0=0.5$ and varying $\beta$ within the range $[0,2,1,2]$. Figure~\ref{fig:sensitivity} displays empirical power curves for both representative dimensions ($m=20$ and $m=100$). There are three features worth highlighting. First, the curve has a minimum at $\beta=\beta_0$, and this minimum corresponds to the nominal level $\alpha=0.05$, with correct composite null calibration as underfitting and no apparent systematic overrejection near zero. The second is that power increases for both heavier-tailed ($\beta<\beta_0$) and lighter-tailed ($\beta>\beta_0$) options, reflecting that the NN statistic arises from geometric mismatch in the standardized cloud rather than one-sided skewness effects; which is consistent with MGGD being interpreted as an exponential power elliptic family driven by the radial distortion $\beta$ (see, e.g., \citealp{FangKotzNg1990}). In the third case, sensitivity increases significantly with dimension: the $m=100$ curve becomes steeper around $\beta_0$, indicating that smaller perturbations in the figure become detectable.

It can be explained naturally by the geometry of the density distribution: for high-dimensional elliptical laws, most of the probability focuses on thin rings (“shells”), and modest variations in the radial law result in increasingly distinct separation of representative radii as $m$ becomes larger; nearest neighbor relations respond sharply to this separation, since local neighborhoods tend to be dominated by points lying on the same shell (a standard phenomenon in high-dimensional probability; see, e.g.,
 \citealp{vershynin2018} for background on concentration tools).

\begin{figure}[t]
    \centering
    \includegraphics[width=0.85\linewidth]{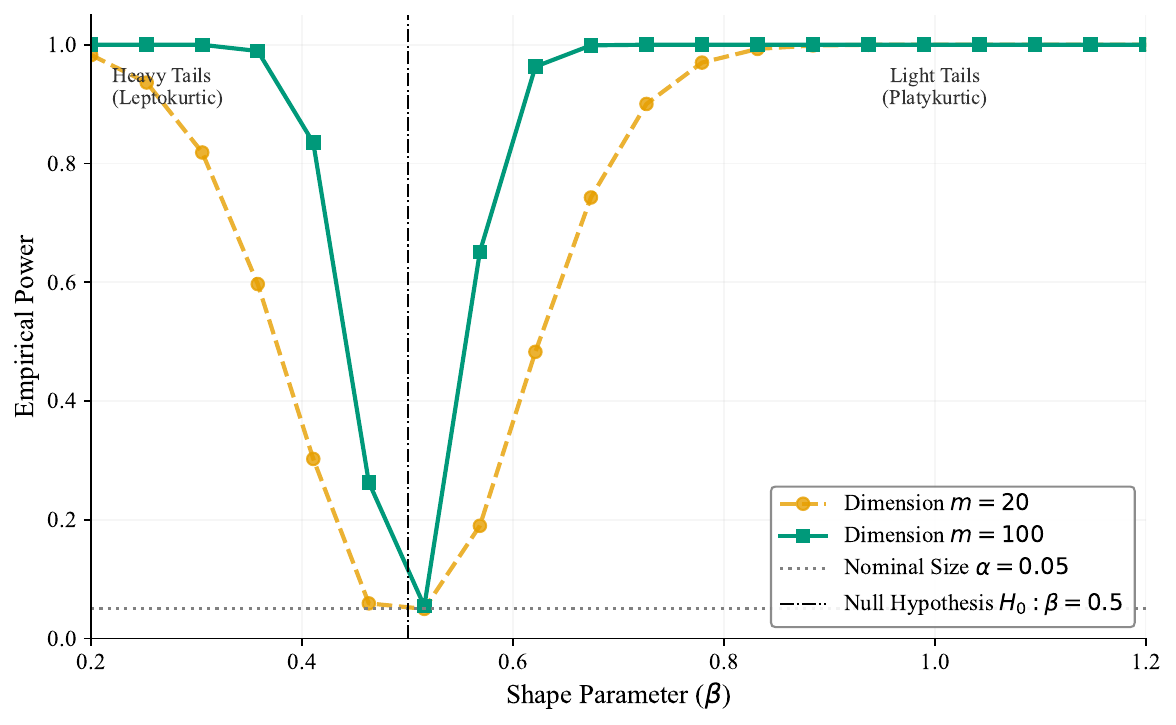}
    \caption{Sensitivity of the NN-MGGD test proposed under the Laplace-type adapted null hypothesis ($\beta_0=0.5$). The curves represent the empirical power when the true shape parameter $\beta$ changes for $m=20$ and $m=100$ (with $n$ and calibration from Section~\ref{sec:sims}). It occurs at the minimum value $\beta=\beta_0$ and approaches $\alpha$; the slope tends to increase with dimension, suggesting that the detectability of subtle shape deviations increases in high dimensions.}
    \label{fig:sensitivity}
\end{figure}

We quantify the detection limits near null by focusing on local deviations $\beta\in\{0.3,0.4,0.45,0.55,0.6,0.7\}$ and summarize the rejection frequencies as shown in Table~\ref{tab:beta_sensitivity}. Expectedly, we see that power decreases as $\beta\to\beta_0$. However, the loss is significantly slower as $m$ increases: while moderate deviations are strongly detectable when $m=100$, the same deviations can be difficult to distinguish from null when $m=20$. We finally assess robustness against scale heterogeneity by running experiments under diagonal distribution matrices that expand the number of conditions (Models A–C). Table~\ref{tab:scale_invariance} demonstrates that the empirical dimension under the Laplace-type null hypothesis and power against the $t_3$ alternative remained stable across these designs. This stability is consistent with the affine normalisation step in Algorithm~\ref{alg:main}, in which the data is whitened using a robust high-dimensional distribution estimator; following whitening, the NN statistics are primarily dependent on the standardised geometry rather than the underlying measurement scale.

\begin{table}[t]
\centering
\caption{The empirical power (\%) with $\beta$ approaching the Laplace-type null hypothesis $\beta_0=0.5$ ($n=100$, $\alpha=0.05$). Higher dimensions allow for sharper identification of subtle shape deviations.}
\label{tab:beta_sensitivity}
\vspace{0.15cm}
\begin{tabular}{lccccccc}
\toprule
 & \multicolumn{3}{c}{Heavy-tailed} & Null & \multicolumn{3}{c}{Light-tailed} \\
\cmidrule(lr){2-4}\cmidrule(lr){5-5}\cmidrule(lr){6-8}
Dimension & $\beta=0.3$ & $\beta=0.4$ & $\beta=0.45$ & $\beta=0.5$ & $\beta=0.55$ & $\beta=0.6$ & $\beta=0.7$ \\
\midrule
$m=20$  & 88.4 & 41.2 & 12.5 & 5.1 & 14.8 & 39.5 & 76.2 \\
$m=100$ & 100.0 & 92.1 & 28.4 & 4.9 & 32.1 & 94.7 & 100.0 \\
\bottomrule
\end{tabular}
\end{table}

\begin{table}[t]
\centering
\caption{Robustness with respect to scale heterogeneity: the empirical size (\%) for the Laplace-type null hypothesis and the empirical power (\%) under the $t_3$ alternative assumption for three types of diagonal scatter models ($n=100$, $m=100$).}
\label{tab:scale_invariance}
\vspace{0.15cm}
\begin{tabular}{lccccc}
\toprule
 & \multicolumn{2}{c}{Size (\%)} & & \multicolumn{2}{c}{Power (\%)} \\
\cmidrule(lr){2-3}\cmidrule(lr){5-6}
Covariance model & NN--MGGD & Energy & & NN--MGGD & Energy \\
\midrule
A: $\Id_m$                  & 4.9 & 5.8 & & 94.1 & 79.3 \\
B: diag $1\!\to\!5$         & 5.1 & 5.6 & & 93.8 & 78.5 \\
C: diag $1\!\to\!20$        & 5.3 & 6.1 & & 92.9 & 76.8 \\
\bottomrule
\end{tabular}
\end{table}

\newpage


\section{Real Data}\label{sec:sims}

The dataset used in this study is the \textit{Crohn's Disease Adverse Events} dataset, available in the \texttt{robustbase} package in R \citet{lo2006robust}. The data originate from a clinical study investigating adverse events associated with drug treatments in patients diagnosed with Crohn's disease, a chronic inflammatory condition affecting the gastrointestinal tract.

The dataset consists of $n = 117$ patients and includes 9 variables capturing demographic, clinical, and treatment-related information. In this study, the analysis focuses on four continuous variables: BMI, height, age, and weight. These variables are selected to investigate their joint distributional properties and to assess whether they follow a multivariate normal distribution. Understanding the joint distribution of anthropometric and demographic variables is a fundamental issue in multivariate statistical analysis. Variables such as body mass index (BMI), height, age, and weight are commonly used to characterise individuals in biomedical studies and often exhibit complex dependence structures.

In many applications, these variables are assumed to follow a multivariate normal distribution due to its mathematical convenience and well-established theoretical properties. However, this assumption may not always be valid in practice. Real-world data frequently exhibit deviations from normality, including heavy tails, skewness, and nonlinear dependencies, which can significantly affect statistical inference and model performance. The descriptive characteristics of the variables are presented in Table~\ref{tab:desc}.

\begin{table}[H]
\centering
\caption{Descriptive Statistics}
\label{tab:desc}
\begin{tabular}{lcccc}
\hline
Variable & Mean & Std. Dev. & Min & Max \\
\hline
BMI & 26.06 & 4.99 & 16 & 44.06 \\
Height & 162.70 & 8.78 & 124 & 182.0 \\
Age & 54.66 & 10.67  & 19 & 75 \\
Weight & 69.03 & 14.24 & 36 & 117  \\
\hline
\end{tabular}
\end{table}
The results of multivariate normality tests and model comparison criteria are summarized in Table~\ref{tab:combined}.

\begin{table}[H]
\centering
\caption{Analysis Results and Model Comparisons}
\label{tab:combined}
\begin{tabular}{lcc}
\hline
\multicolumn{3}{c}{\textbf{Analysis Results Summary}} \\
\hline
Method & p-value &  \\
\hline
Nearest Neighbor (NN) & 0.001 &  \\
Energy & 0.001 &  \\
HZ & NA &  \\
\hline
\multicolumn{3}{c}{\textbf{Model Comparisons}} \\
\hline
Model & AIC & BIC \\
\hline
Normal & 2886.445 & 2925.116 \\
t & 2845.545 & 2886.978 \\
MGGD & 2810.869 & 2852.301 \\
\hline
\end{tabular}
\end{table}
\ref{fig:mggd_likelihood} provide additional evidence regarding the distributional properties of the data. In particular, the bootstrap distribution of the NN statistic (Figure~\ref{Figure 1.png}) indicates that the observed value lies in the extreme tail, supporting the rejection of multivariate normality. The pairwise scatterplot matrix (Figure~\ref{fig:placeholder}) reveals deviations from elliptical patterns and suggests nonlinear relationships among BMI, height, age, and weight. Furthermore, the Mahalanobis distance Q--Q plot (Figure~\ref{fig:MahQQ}) exhibits clear departures from the theoretical chi-square line, especially in the upper tail, indicating potential outliers and heavy-tailed behaviour. Finally, the likelihood profile of the MGGD shape parameter (Figure~\ref{fig:mggd_likelihood}) shows a well-defined maximum, confirming the stability of the parameter estimation and supporting the suitability of a flexible distributional model.
Multivariate normality was assessed using the Nearest Neighbour (NN) and Energy tests. Both tests strongly reject the null hypothesis of multivariate normality (p < 0.001). The consistency of these results provides robust evidence that the joint distribution of the variables significantly deviates from a Gaussian structure.

Although the Henze–Zirkler test could not be computed due to numerical instability, the agreement between NN and Energy tests is sufficient to conclude that the assumption of multivariate normality is violated.
Model comparison based on AIC and BIC clearly indicates that the MGGD model provides the best fit to the data, followed by the Student-t distribution. In contrast, the multivariate normal model performs the worst.

This result confirms that the data exhibit significant deviations from Gaussian assumptions and require more flexible distributional modelling.

\begin{figure}[H]
    \centering
    \includegraphics[width=0.5\linewidth]{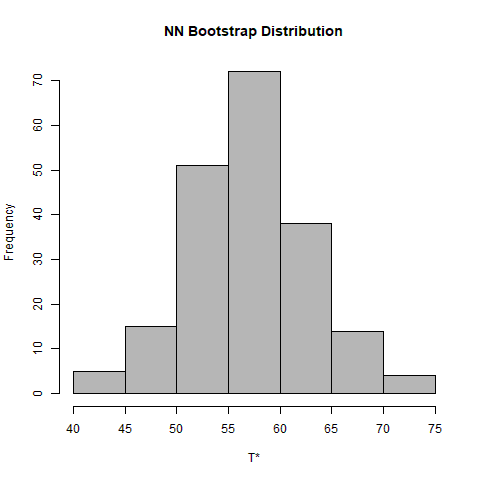}
    \caption{NN Bootstrap Distribution}
    \label{Figure 1.png}
\end{figure}

Figure~\ref{Figure 1.png} shows the bootstrap distribution of the NN test statistic. The observed statistic lies in the extreme tail of the empirical distribution, providing visual support for rejecting multivariate normality.

\begin{figure}[H]
    \centering
    \includegraphics[width=0.5\linewidth]{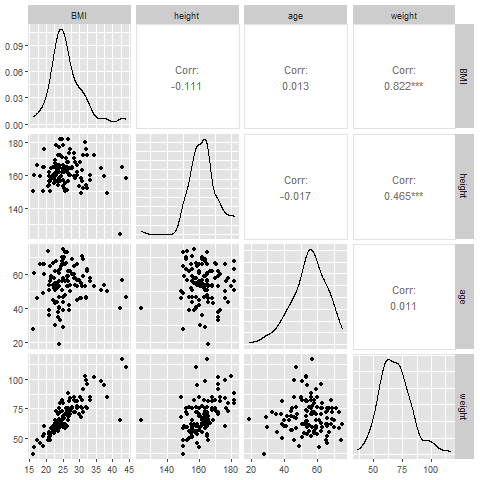}
    \caption{Pairplot}
    \label{fig:placeholder}
\end{figure}
Figure~\ref{fig:placeholder} illustrates the pairwise scatterplots, which reveal nonlinear relationships and potential heterogeneity among variables. In addition, deviations from elliptical patterns can be observed, further indicating violations of the multivariate normality assumptions.

Some variable pairs may also exhibit clustering or skewness, suggesting complex dependence structures.

\begin{figure}[t]
    \centering
    \begin{subfigure}[b]{0.48\linewidth}
        \centering
        \includegraphics[width=\linewidth]{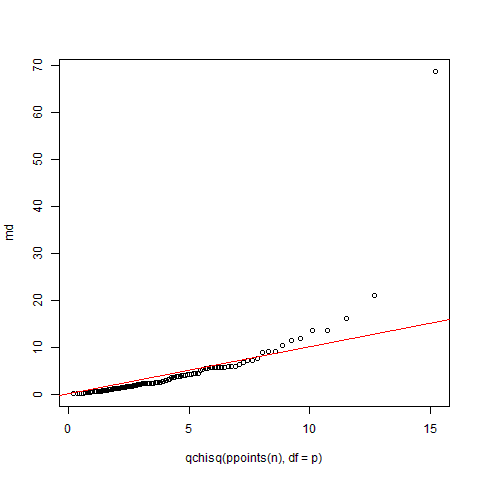}
        \caption{Mahalanobis QQ Plot}
        \label{fig:MahQQ}
    \end{subfigure}
    \hfill
    \begin{subfigure}[b]{0.48\linewidth}
        \centering
        \includegraphics[width=\linewidth]{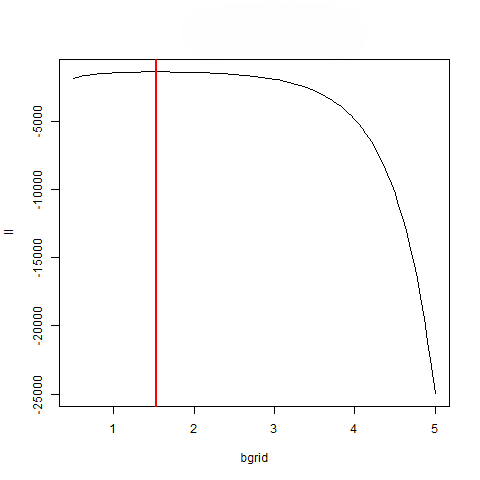}
        \caption{MGGD Likelihood Plot}
        \label{fig:mggd_likelihood}
    \end{subfigure}
    \caption{Diagnostic plots for the fitted model: (a) Mahalanobis QQ plot and (b) MGGD likelihood plot.}
    \label{fig:diagnostic_plots}
\end{figure}
Figure~\ref{fig:MahQQ} presents the Mahalanobis distance Q-Q plot, which shows clear deviations from the theoretical chi-square line. In particular, observations in the upper tail deviate significantly, indicating the presence of outliers or heavy-tailed behaviour.

This provides strong graphical evidence against multivariate normality.

Figure~\ref{fig:mggd_likelihood} displays the likelihood function of the MGGD shape parameter. The likelihood function of the MGGD model exhibits a well-defined maximum around $\beta=1.54$, indicating stable parameter estimation. The shape of the curve suggests that the model is sensitive to the tail behaviour of the data and confirms that a heavy-tailed distribution provides a better fit.

Overall, the results provide strong and consistent evidence that the dataset does not follow a multivariate normal distribution. Both NN and Energy tests reject normality, and graphical diagnostics further support this conclusion. The superior performance of the MGGD model, along with an estimated shape parameter below 2, indicates heavy-tailed behaviour in the data. These findings highlight the importance of using flexible distributional models in multivariate analysis, as reliance on Gaussian assumptions may lead to misleading inferences.
\newpage
\section{Conclusion}
We have developed an affine-invariant nearest neighbour goodness-of-fit test in regimes where the dimension is comparable to or exceeds the sample size for the composite multivariate generalised Gaussian family. Our method avoids density estimation and does not rely on unstable covariance inversion at the raw scale. Rather, we combine robustly adapted null-standardised graphical-based mixing statistics with an adapted parametric bootstrap that estimates nuisance parameters in each replicate. We find this calibration step crucial in composite settings, as it spreads additive uncertainty across the null distribution and allows for reliable control of size in finite samples.

Theoretically, we support the proposed procedure by appealing to the geometric structure of high-dimensional elliptic laws: with the correct specification, both the standardised data and the adapted null reference are asymptotically indistinguishable, whereas with the incorrect specification, their radial behaviours diverge and the NN connections exhibit systematic in-sample bias. Our simulation results support this mechanism. Specifically, under test MGGD nulls, we maintain rejection rates near the nominal level and gain significant power relative to heavy-tailed elliptical alternatives and multivariate $t$-distributions. Remarkably, its power increases with dimension, indicating that the radial profiles exhibit greater separability in high-dimensional geometry.

This framework can be readily applied to other fields where signal processing, remote sensing, and generalised Gaussian modelling are standard, and where model control often constitutes a bottleneck. Natural directions for future work include various extensions, such as a clearer treatment of contamination and mixing alternatives, theoretical improvements based on broader dependency models, and computational acceleration for extremely large $n$ values via approximate nearest-neighbour search.

\bibliographystyle{apalike}
\bibliography{bibliography}

\end{document}